\documentclass[12pt]{article}
\usepackage{myextension}
\usepackage{tocloft}
\title{A new proof of Delahan’s induced‑universality result}
\author{Jonathan CHAPPELON\footnote{E-mail address: \href{mailto:jonathan.chappelon@umontpellier.fr}{jonathan.chappelon@umontpellier.fr}}}
\affil{IMAG, Université de Montpellier, CNRS, Montpellier, France}
\date{\today}

\begin{document}
\maketitle
\begin{abstract}
We give a short and self-contained proof of Delahan’s theorem stating that every simple graph on $n$ vertices occurs as an induced subgraph of a Steinhaus graph on $\tfrac{n(n-1)}{2}+1$ vertices. This new proof is obtained by considering the notion of generating index sets for Steinhaus triangles.
\end{abstract}

\section{Introduction}

All along this paper, $n$ is a non-negative integer. The set of non-negative integers is denoted by $\N$. Let $t_n$ denote the $n$\textsuperscript{th} triangular number
$$
t_{n}=\sum_{i=0}^{n}i=\frac{n(n+1)}{2}
$$
and let $\Tn{n}$ be the triangle of integers
$$
\Tn{n} = \left\{ (i,j)\in\N^2\ \middle|\ i+j<n\right\}.
$$

A {\em binary Steinhaus triangle} (or {\em Steinhaus triangle} for short) of {\em size} $n$ is a down-pointing triangle $\left(a_{i,j}\right)_{(i,j)\in\Tn{n}}$ of $t_{n}$ zeroes and ones satisfying the same local rule as the standard Pascal triangle modulo $2$, that is,
\begin{equation}\label{eq01}
a_{i,j} \equiv a_{i-1,j} + a_{i-1,j+1} \pmod{2},
\end{equation}
for all $(i,j)\in\Tn{n}$ such that $i\ge1$. Note that $(0)$ and $(1)$ are the Steinhaus triangles of size $1$ and $\emptyset$ is the Steinhaus triangle of size $0$. An example of a Steinhaus triangle of size $7$ is depicted in Figure~\ref{fig01}

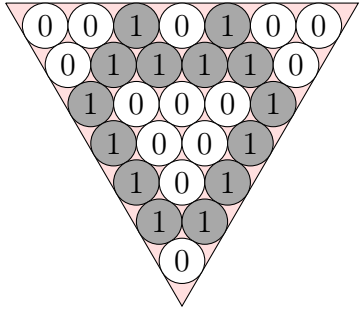
\begin{figure}[htbp]
\begin{center}
\begin{tikzpicture}[scale=0.3]
\figST{{0,0,1,0,1,0,0}}{black}{col0}{col1}{pink!50}
\end{tikzpicture}
\end{center}
\caption{The Steinhaus triangle $\ST{(0010100)}$ of size $7$}\label{fig01}
\end{figure}

It is clear that a Steinhaus triangle $\left(a_{i,j}\right)_{(i,j)\in\Tn{n}}$ is completely determined by its first row $\left(a_{0,j}\right)_{j=0}^{n-1}$. Indeed, by induction on $i$ and using \eqref{eq01}, we obtain
\begin{equation}\label{eq02}
a_{i,j} \equiv \sum_{k=0}^{i}\binom{i}{k}a_{0,j+k}\pmod{2},
\end{equation}
for all $(i,j)\in\Tn{n}$, where the binomial coefficient $\binom{a}{b}$ is the coefficient of the monomial $X^b$ in the expansion of $(1+X)^a$, for all non-negative integers $a$ and $b$ such that $b\le a$. In the sequel, the Steinhaus triangle whose first row is the sequence $S$ is denoted by $\ST{S}$. The Steinhaus triangle in Figure~\ref{fig01} is then $\ST{(0010100)}$.

Since the set $\setST{n}$ of binary Steinhaus triangles of size $n$ is closed under addition modulo $2$, it follows that $\setST{n}$ is a vector space over $\Zn{2}$. Moreover, since a Steinhaus triangle is uniquely determined by its first row, the dimension of $\setST{n}$ is $n$.

Beyond this elementary definition, Steinhaus triangles encode a surprisingly rich combinatorial structure. In particular, a classical line of research, originating from work of Steinhaus \cite{Steinhaus:1964aa} and Molluzzo \cite{Molluzzo:1978aa}, asks for conditions under which a Steinhaus triangle is {\em balanced}, that is, each symbol (or residue class, in a modular extension) occurs with the same multiplicity. In its binary form, this leads to questions about the distribution of zeroes and ones across the triangle \cite{Harborth:1972aa,Eliahou:2004aa,Chappelon2017a}. In a more general form, one studies Steinhaus triangles over $\Zn{m}$ and asks for the existence (or explicit construction) of balanced configurations for admissible sizes \cite{Chappelon2008,Chappelon2011,Chappelon2025}. This has fueled a sustained interest in structural and algebraic aspects of Steinhaus triangles (symmetries, linear subspaces, modular lifts, ...) as in \cite{Chappelon2019}, providing a broad context for the present work.

A {\em simple graph} is a pair $G=(V,E)$ where $V$ is a finite vertex set and $E\subset\binom{V}{2}$ is a set of unordered pairs of distinct vertices (no loops, no multiple edges). For $W\subset V$, the {\em induced subgraph} $G[W]$ is the graph with vertex set $W$ and
$$
E(G[W]) = \left\{ \{u,v\}\in E\ \middle|\ u,v\in W\right\}.
$$
In other words, it contains all edges of $G$ with both endpoints in $W$, and no others.

A {\em Steinhaus graph} of {\em order} $n\ge 1$ is a simple graph whose adjacency matrix has an upper-triangular part which is a Steinhaus triangle of size $n-1$. For any sequence $S=\left(a_1,a_2,\ldots,a_{n-1}\right)$ of zeroes and ones of length $n-1$, its associated Steinhaus graph $\Sgraph{S}$ is the simple graph of order $n$ whose adjacency matrix $\Smat{S}=\left(a_{i,j}\right)_{1\le i,j\le n}$ verifies
\begin{enumerate}[i)]
\item
$a_{i,j}=a_{j,i}$, for all $i,j\in\{1,2,\ldots,n\}$, (symmetry)
\item
$a_{i,i}=0$, for all $i\in\{1,2,\ldots,n\}$, (diagonal of zeroes)
\item
$a_{1,j}=a_{j-1}$, for all $j\in\{2,3,\ldots,n\}$, (sequence $S$)
\item
$a_{i,j}\equiv a_{i-1,j-1}+a_{i-1,j}\pmod{2}$, for all integers $i,j$ such that $2\le i<j\le n$, (local rule of $\ST{S}$).
\end{enumerate}
For example, for $S=(0010100)$, the Steinhaus graph $\Sgraph{S}$ and its adjacency matrix $\Smat{S}$ are depicted in Figure~\ref{fig02}.

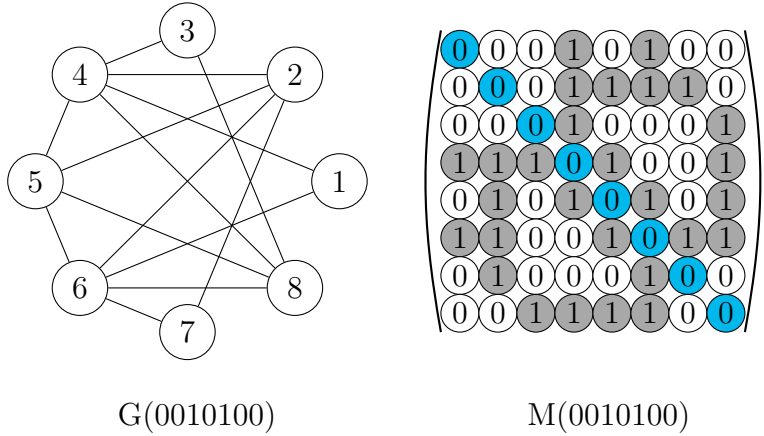
\begin{figure}[htbp]
\centerline{
\begin{tabular}{cc}
\begin{minipage}{5cm}
\begin{tikzpicture}[scale=0.5]
\figSG{{0,0,1,0,1,0,0}}{black}{white}
\end{tikzpicture}
\end{minipage}
&
\begin{minipage}{5cm}
\begin{tikzpicture}[scale=0.25]
\figSM{{0,0,1,0,1,0,0}}{black}{col0}{col1}{col3}
\end{tikzpicture}
\end{minipage}
\\ \ \\
$\Sgraph{0010100}$ & $\Smat{0010100}$ \\
\end{tabular}}
\caption{The Steinhaus graph $\Sgraph{0010100}$ and its adjacency matrix $\Smat{0010100}$}\label{fig02}
\end{figure}

The set of Steinhaus graphs of order $n$ is denoted by $\setSG{n}$. It is clear that there is a natural correspondence between $\setSG{n}$ and $\setST{n-1}$. Therefore, the set $\setSG{n}$ is a vector space over $\Zn{2}$ of dimension $n-1$.

A classical problem on Steinhaus graphs is to study those having certain graphical
properties such as bipartition \cite{Dymacek:1986aa,Dymacek:1995aa,Chang:1999aa}, planarity \cite{Dymacek:2000aa} or regularity \cite{Dymacek:1979aa,Bailey:1988aa,Augier:2008aa,Chappelon2009}. A survey on Steinhaus graphs can be found in \cite{Dymacek:1999aa}.

While Molluzzo’s problem concerns distributional properties inside Steinhaus triangles, the graph-theoretic perspective opens up a different class of questions about \emph{embeddings} and \emph{universality} within the family of Steinhaus graphs. A landmark theorem in this direction is due to Delahan: for each $n\ge 1$, every simple graph on $n$ vertices appears as an induced subgraph of a Steinhaus graph on $t_{n-1}+1$ vertices. Equivalently, the family of Steinhaus graphs of order $t_{n-1}+1$ is $n$-induced-universal.

For any positive integer $n$, the set of all simple graphs on the vertex set $\{1,2,\ldots,n\}$ is denoted by $\mathcal{G}_n$. In \cite{Delahan:1998aa}, it was proved that 
there exists an isomorphism between $\mathcal{G}_n$ and $\setSG{t_{n-1}+1}$.

\begin{thm}[Delahan \cite{Delahan:1998aa}]\label{thm1}
Let $n$ be a positive integer. We consider the $n$-subset
$$
W_{n} = \left\{ t_i+1\ \middle|\ i\in\{0,1,\ldots,n-1\} \right\}
$$
of $\left\{1,2,\ldots,t_{n-1}+1\right\}$. Then, the linear map
$$
\begin{array}{ccc}
\setSG{t_{n-1}+1} & \longrightarrow & \mathcal{G}_n \\
G & \longmapsto & G[W_n] \\
\end{array}
$$
is an isomorphism. Therefore, any simple graph of order $n$ is isomorphic to an induced subgraph of a Steinhaus graph of order $t_{n-1}+1$.
\end{thm}

Delahan’s theorem shows that Steinhaus graphs are, in a precise induced sense, "large enough" to contain all $n$-vertex graphs. The goal of this paper is to give a new, short, and self-contained proof of Delahan’s induced-universality theorem by considering the notion of generating index sets of Steinhaus triangles.

This paper is organized as follows. Section~2 introduces generating index sets for Steinhaus triangles and characterizes them by invertibility modulo $2$. In Section~3, we study certain binomial minors in particular and prove that the relevant determinants are odd. Finally, in Section~4, by considering the results of Sections~2 and 3, Delahan’s theorem is proved.

\section{Generating index sets}

A subset $A$ of $\Tn{n}$ is said to be a {\em generating index set} of $\setST{n}$ if the knowledge of the values $a_{i,j}$, for all $(i,j)\in A$, uniquely determines the whole Steinhaus triangle $\left(a_{i,j}\right)_{(i,j)\in\Tn{n}}$ where the size of $A$ is minimal; i.e., a subset $A$ that makes the following linear map an isomorphism: 
$$
\pi_A :
\begin{array}[t]{ccc}
\setST{n} & \longrightarrow & \Z_2^{|A|} \\
\left(a_{i,j}\right)_{(i,j)\in\Tn{n}} & \longmapsto & \left(a_{i,j}\right)_{(i,j)\in A} \\
\end{array}
$$
where $\Z_2 = \left\{0,1\right\}$. Since $\dim\setST{n}=n$, we deduce that the cardinality of a generating index set of $\setST{n}$ is always $n$. From \eqref{eq02}, it is clear that the set of top row indices of a Steinhaus triangle of size $n$, that is,
$$
A_T = \left\{(0,0),(0,1),\ldots,(0,n-1)\right\},
$$
is a generating index set of $\setST{n}$. Note that $\pi_{A_T}^{-1}\!\left(S\right)=\ST{S}$, for all $S\in\Z_2^{n}$. It follows that the set $A$ is a generating index set of $\setST{n}$ if and only if the linear map
$$
\pi_A\circ\pi_{A_T}^{-1} : \Z_2^{n} \longrightarrow \Z_2^{n}
$$
is an automorphism. For instance, the $16$ generating index sets of $\setST{3}$ (4 up to the action of the dihedral group $D_3$) are depicted in Figure~\ref{fig03}, where a disk is either black if its position is in the generating index set or white otherwise. Moreover, the four $3$-subsets of $\Tn{3}$ that are not generating index sets of $\setST{3}$ are given in Figure~\ref{fig04}. 

\begin{figure}[htbp]
\centerline{\includegraphics[width=\textwidth]{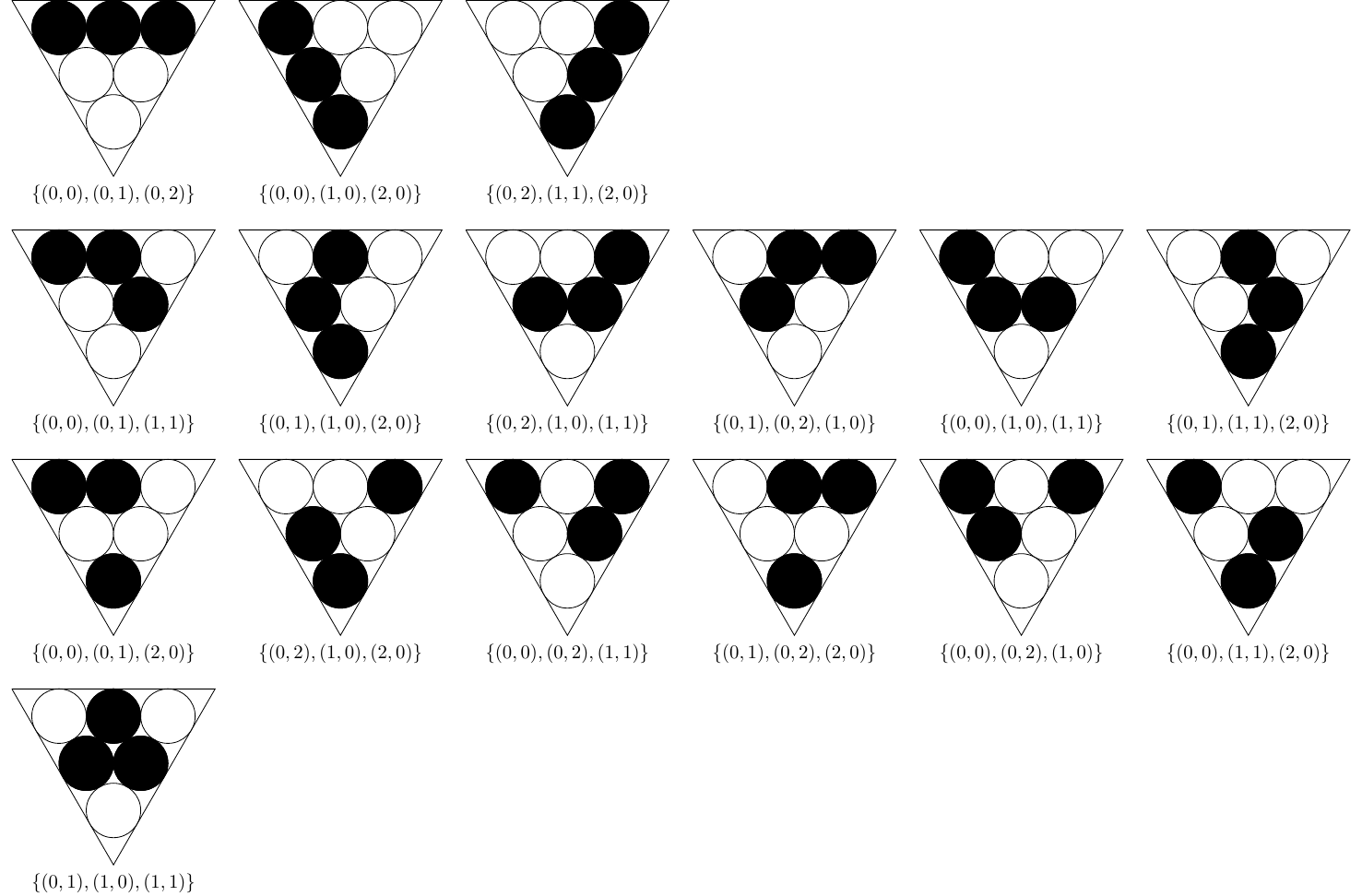}}
\caption{Generating index sets of $\setST{3}$}\label{fig03}
\end{figure}

\begin{figure}[htbp]
\centerline{\includegraphics[width=0.75\textwidth]{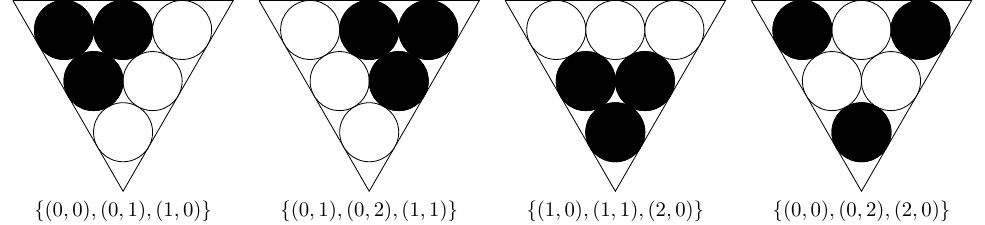}}
\caption{$3$-subsets of $\Tn{3}$ that are not generating index sets of $\setST{3}$}\label{fig04}
\end{figure}

Since the sets of right side indices
$$
A_R = \left\{ (0,n-1),(1,n-2),\ldots,(n-1,0) \right\}
$$
and left side indices
$$
A_L = \left\{ (0,0),(1,0),\ldots,(n-1,0) \right\}
$$
of a Steinhaus triangle $\ST{}$ of size $n$ can be seen as the top row indices of the $120^\circ$ and $240^\circ$ rotations of $\ST{}$, respectively, it follows that $A_R$ and $A_L$ are generating index sets of $\setST{n}$ too. Therefore, each element of a Steinhaus triangle can be expressed as a function of the terms of its first row, of its right side or of its left side.

For any non-negative integers $a$ and $b$ such that $b\le a$, the binomial coefficient $\binom{a}{b}$ corresponds to the number of ways to choose $b$ elements in a set of $a$ elements. Here, we extend this notation by supposing that $\binom{a}{b}=0$, for all integers $b$ such that $b<0$ or $b>a$. For this generalization, the Pascal identity
$$
\binom{a}{b} = \binom{a-1}{b-1}+\binom{a-1}{b}
$$
holds, for all positive integers $a$ and all integers $b$.

\begin{lem}\label{lem1}
Let $\left(a_{i,j}\right)_{(i,j)\in\Tn{n}}$ be a binary Steinhaus triangle of size $n$. Then, we have
$$
a_{i,j} \equiv \sum_{k=0}^{n-1}\binom{i}{k-j}a_{0,k} \equiv \sum_{k=0}^{n-1}\binom{n-1-(i+j)}{k-i}a_{k,n-1-k} \equiv \sum_{k=0}^{n-1}\binom{j}{k-i}a_{k,0} \pmod{2}
$$
for all $(i,j)\in\Tn{n}$.
\end{lem}

\begin{proof}
By induction on $i$ using \eqref{eq01}.
\end{proof}

\begin{prop}\label{prop3}
Let $A=\left\{(i_0,j_0),(i_1,j_1),\ldots,(i_{n-1},j_{n-1})\right\}$ be a $n$-subset of $\Tn{n}$. Then, the set $A$ is a generating index set of $\setST{n}$ if and only if the matrix
$$
\Md{A} = \left( \binom{i_k}{\ell-j_k} \right)_{0\le k,\ell\le n-1}
$$
is invertible modulo $2$.
\end{prop}

\begin{proof}
From Lemma~\ref{lem1}, we know that
$$
a_{i_k,j_k} \equiv \sum_{\ell=0}^{n-1}\binom{i_k}{\ell-j_k}a_{0,\ell}\pmod{2},
$$
for all $k\in\{0,1,\ldots,n-1\}$. It follows that
$$
\pi_A\circ\pi_{A_T}^{-1}(S) \equiv S.{\Md{A}}^{\intercal} \pmod{2},
$$
for all $S\in\Z_2^n$. Finally, the linear map $\pi_A\circ\pi_{A_T}^{-1}$ is an automorphism if and only if the matrix $\Md{A}$ is invertible modulo $2$.
\end{proof}

The notion of generating index sets and the result of Proposition~\ref{prop3} appear in a more general context in \cite{Barbe:2004aa,Barbe:2009aa,Chappelon2019}. It is also proved that the set of generating index sets of $\setST{n}$ define a matroid called the Pascal matroid modulo 2 \cite{Barbe:2004aa,Barbe:2009aa}.

\section{Minors of binomial matrix}

In this section, we consider the infinite binomial matrix
$$
\matr{\mathrm{B}} = \left(\binom{i}{j}\right)_{(i,j)\in\N^2}.
$$
The {\em submatrix} of order $n$ corresponding to rows $0\le a_1<a_2<\cdots<a_n$ and columns $0\le b_1<b_2<\cdots<b_n$ is denoted by
$$
\matbin{a_1,a_2,\ldots,a_n}{b_1,b_2,\ldots,b_n} = \left(\binom{a_i}{b_j}\right)_{1\le i,j\le n}.
$$

\begin{prop}\label{prop1}
For any increasing sequence of $n$ non-negative integers $0\le a_1<a_2<\cdots<a_n$, we have
$$
\det\matbin{a_1,a_2,\ldots,a_n}{0,1,\ldots,n-1} = \frac{\Delta\!\left(a_1,a_2,\ldots,a_n\right)}{\prod_{k=0}^{n-1}k!},
$$
where $\Delta\!\left(a_1,a_2,\ldots,a_n\right)$ is the product of the $t_{n-1}$ differences
$$
\Delta\!\left(a_1,a_2,\ldots,a_n\right) = \prod_{1\le i<j\le n}(a_j-a_i).
$$
\end{prop}

\begin{proof}
First, by definition, we have
$$
\det\matbin{a_1,a_2,\ldots,a_n}{0,1,\ldots,n-1} = \det\!\left(\binom{a_i}{j-1}\right)_{1\le i,j\le n} =
\renewcommand{\arraystretch}{1.5}
\begin{vmatrix}
\binom{a_1}{0} & \binom{a_1}{1} & \cdots & \binom{a_1}{n-1} \\
\binom{a_2}{0} & \binom{a_2}{1} & \cdots & \binom{a_2}{n-1} \\
\vdots & \vdots & \ddots & \vdots \\
\binom{a_n}{0} & \binom{a_n}{1} & \cdots & \binom{a_n}{n-1}
\end{vmatrix}.
\renewcommand{\arraystretch}{1}
$$
Since, for all $i,j\in\{1,2,\ldots,n\}$, we have
$$
\binom{a_i}{j-1} = \frac{(a_i)_{j-1}}{(j-1)!},
$$
where $(a_i)_{j-1}$ is the falling factorial
$$
(a_i)_{j-1} = \prod_{k=0}^{j-2}(a_i-k),
$$
we obtain that
$$
\det\matbin{a_1,a_2,\ldots,a_n}{0,1,\ldots,n-1}
\begin{array}[t]{l}
 = \displaystyle\det\!\left(\frac{(a_i)_{j-1}}{(j-1)!}\right)_{1\le i,j\le n} =
\renewcommand{\arraystretch}{1.5}
\begin{vmatrix}
\frac{(a_1)_0}{0!} & \frac{(a_1)_1}{1!} & \cdots & \frac{(a_1)_{n-1}}{(n-1)!} \\
\frac{(a_2)_0}{0!} & \frac{(a_2)_1}{1!} & \cdots & \frac{(a_2)_{n-1}}{(n-1)!} \\
\vdots & \vdots & \ddots & \vdots \\
\frac{(a_n)_0}{0!} & \frac{(a_n)_1}{1!} & \cdots & \frac{(a_n)_{n-1}}{(n-1)!}
\end{vmatrix} \renewcommand{\arraystretch}{1}
\\ \ \\
 = \displaystyle\frac{1}{\prod_{k=0}^{n-1}k!}\det\!\left((a_i)_{j-1}\right)_{1\le i,j\le n}.
\end{array}
$$
Moreover, since, for all $i,j\in\{1,2,\ldots,n\}$, we have
$$
(a_i)_{j-1} = a_i(a_i-1)\cdots(a_i-j+2) = {a_i}^{j-1} + \sum_{k=0}^{j-2}s(j-1,k){a_i}^k,
$$
where $s(j-1,k)$ are Stirling numbers of the first kind, we deduce that
$$
\det\!\left((a_i)_{j-1}\right)_{1\le i,j\le n} =
\renewcommand{\arraystretch}{1.5}
\begin{vmatrix}
(a_1)_0 & (a_1)_1 & \cdots & (a_1)_{n-1} \\
(a_2)_0 & (a_2)_1 & \cdots & (a_2)_{n-1} \\
\vdots & \vdots & \ddots & \vdots \\
(a_n)_0 & (a_n)_1 & \cdots & (a_n)_{n-1}
\end{vmatrix} \renewcommand{\arraystretch}{1}
= \begin{vmatrix}
1 & a_1 & {a_1}^2 & \cdots & {a_1}^{n-1} \\
1 & a_2 & {a_2}^2 & \cdots & {a_2}^{n-1} \\
\vdots & \vdots & \vdots & \ddots & \vdots \\
1 & a_n & {a_n}^2 & \cdots & {a_n}^{n-1}
\end{vmatrix}.
$$
Since the determinant of this Vandermonde matrix corresponds to $\Delta\!\left(a_1,a_2,\ldots,a_n\right)$, we conclude that
$$
\det\matbin{a_1,a_2,\ldots,a_n}{0,1,\ldots,n-1} = \frac{\Delta\!\left(a_1,a_2,\ldots,a_n\right)}{\prod_{k=0}^{n-1}k!},
$$
as announced.
\end{proof}

Beyond the Vandermonde type evaluation in Proposition~\ref{prop1}, there is a rich body of results on binomial minors. In particular, Gessel and Viennot provided a combinatorial interpretation of any minor of the binomial coefficient matrix in terms of non-intersecting paths, with connections to Young tableaux and hook length formulas \cite{Gessel1985}.

For the particular case where $(a_1,a_2,\ldots,a_n)$ is the sequence of the first $n$ triangular numbers, i.e., $a_i=t_{i-1}$ for all $i\in\{1,2,\ldots,n\}$, we obtain the following

\begin{prop}\label{prop2}
For the first $n$ triangular numbers, we have
$$
\det\matbin{t_0,t_1,\ldots,t_{n-1}}{0,1,\ldots,n-1} = \prod_{i=1}^{n-1}{(2i-1)}^{n-i}.
$$
\end{prop}

\begin{proof}
First, from definition of the product of differences, we have
$$
\Delta\!\left(t_0,t_1,\ldots,t_{n-1}\right) = \prod_{0\le i<j\le n-1}\left(t_j-t_i\right).
$$
For any non-negative integers $i$ and $j$, it is clear that
$$
t_j - t_i = \frac{j(j+1)}{2} - \frac{i(i+1)}{2} = \frac{(j-i)(j+i+1)}{2}.
$$
It follows that
$$
\Delta\!\left(t_0,t_1,\ldots,t_{n-1}\right) 
\begin{array}[t]{l}
= \displaystyle\prod_{j=1}^{n-1}\prod_{i=0}^{j-1}\frac{(j-i)(j+i+1)}{2} = \prod_{j=1}^{n-1}\frac{(2j)!}{2^j} = \prod_{j=1}^{n-1}\prod_{i=1}^{j}\frac{2i(2i-1)}{2} \\
= \displaystyle\prod_{j=1}^{n-1}\prod_{i=1}^{j}i(2i-1) = \prod_{j=1}^{n-1}j!\prod_{i=1}^{j}(2i-1) = \left(\prod_{j=1}^{n-1}j!\right)\prod_{j=1}^{n-1}\prod_{i=1}^{j}(2i-1) \\
= \displaystyle\left(\prod_{j=1}^{n-1}j!\right)\prod_{i=1}^{n-1}\prod_{j=i}^{n-1}(2i-1) = \left(\prod_{j=1}^{n-1}j!\right)\prod_{i=1}^{n-1}{(2i-1)}^{n-i}.
\end{array}
$$
Using Proposition~\ref{prop1}, we conclude that
$$
\det\matbin{t_0,t_1,\ldots,t_{n-1}}{0,1,\ldots,n-1} = \frac{\Delta\!\left(t_0,t_1,\ldots,t_{n-1}\right)}{\prod_{k=0}^{n-1}k!} = \prod_{i=1}^{n-1}{(2i-1)}^{n-i},
$$
as announced.
\end{proof}

Since their determinants are odd by Proposition~\ref{prop2}, we then know that the binomial matrices of the form
$$
\matbin{t_0,t_1,\ldots,t_{n-1}}{0,1,\ldots,n-1}
$$
are all invertible modulo $2$, for all positive integers $n$.

\section{Proof of Delahan's result}

Using the notion of generating index sets, it is easy to see that Theorem~\ref{thm1} is a direct corollary of the following

\begin{thm}\label{thm2}
The set
$$
A_n = \left\{ (t_s,t_{r+1}-t_s-1) \ \middle|\ 0\le s\le r\le n-2 \right\}
$$
is a generating index set of $\setST{t_{n-1}}$.
\end{thm}

First, the set $A_n$ can be seen as
$$
A_n=\left\{ (i_k,j_k)\ \middle| \ k\in\{0,1,\ldots,t_{n-1}-1\} \right\},
$$
where
$$
i_k = t_s \quad\text{and}\quad j_k=t_{r+1}-t_s-1,
$$
when $k=t_r+s$ with $0\le s\le r\le n-2$. We then consider the square matrix $\Md{A_n}$ of size $t_{n-1}$ defined by
$$
\Md{A_n} = \left( \binom{i_k}{\ell-j_k} \right)_{0\le k,\ell\le t_{n-1}-1}.
$$

The proof of Theorem~\ref{thm2} is based on the following

\begin{lem}\label{lem2}
The matrix $\Md{A_n}$ is a block lower-triangular matrix with square diagonal blocks of size $1,2,\ldots,n-1$
$$
\matbin{t_0,t_1,\ldots,t_r}{r,r-1,\ldots,0},
$$
for all $r\in\{0,1,\ldots,n-2\}$.
\end{lem}

For instance, the block lower-triangular structure of $\Md{A_5}$ is depicted in Figure~\ref{fig05}.

\begin{figure}[htbp]
\begin{center}
\begin{tikzpicture}[scale=0.5]
\draw (-1,-1) -- (-1,1)-- (1,1) -- (1,-1) -- (-1,-1);
\draw (0,0) node {\tiny $\matbin{t_0}{0}$};
\draw (1,-1) -- (5,-1) -- (5,-5) -- (1,-5) -- (1,-1);
\draw (3,-3) node {$\matbin{t_0,t_1}{1,0}$};
\draw (5,-5) -- (11,-5) -- (11,-11) -- (5,-11) -- (5,-5);
\draw (8,-8) node {$\matbin{t_0,t_1,t_2}{2,1,0}$};
\draw (11,-11) -- (19,-11) -- (19,-19) -- (11,-19) -- (11,-11);
\draw (15,-15) node {$\matbin{t_0,t_1,t_2,t_3}{3,2,1,0}$};
\draw (-1,1) -- (19,1) -- (19,-19) -- (-1,-19) -- (-1,1);
\draw (3,0) node {$\matr{0}$};
\draw (8,-2) node {$\matr{0}$};
\draw (15,-5) node {$\matr{0}$};
\draw (0,-3) node {$\matr{*}$};
\draw (2,-8) node {$\matr{*}$};
\draw (5,-15) node {$\matr{*}$};
\end{tikzpicture}
\end{center}
\caption{Structure of the matrix $\Md{A_5}$}\label{fig05}
\end{figure}
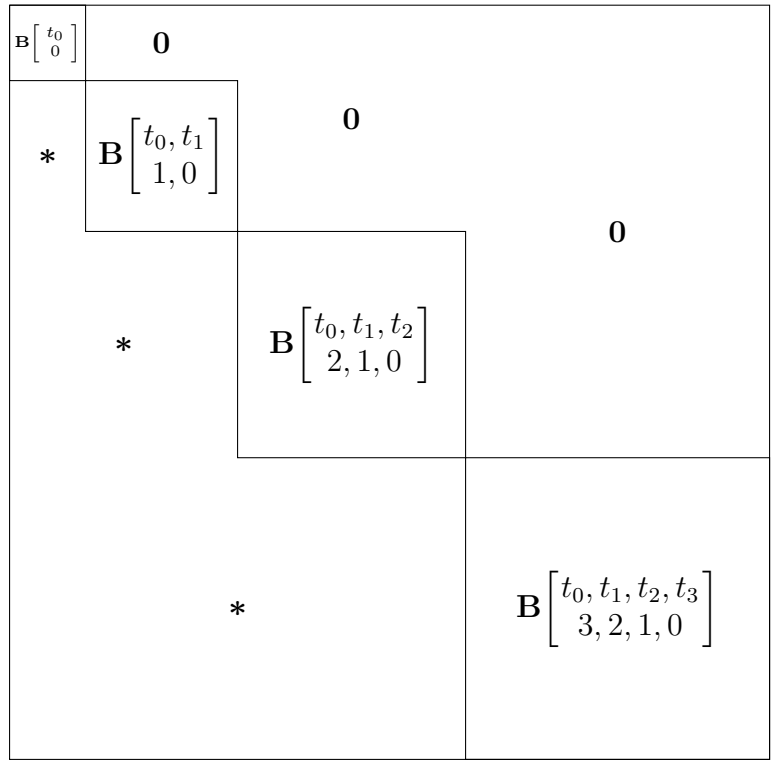

\begin{proof}
Let $r\in\{0,1,\ldots,n-2\}$. We can see that the $r+1$ rows
$$
R_{t_r},\ R_{t_r+1},\ \ldots,\ R_{t_r+r}
$$
have a particular form. For $k=t_r+s$ with $s\in\{0,1,\ldots,r\}$, we know that $i_k=t_s$ and $j_k=t_{r+1}-t_s-1$. Then,
$$
R_{t_r+s} = \left(\binom{t_s}{\ell-t_{r+1}+t_s+1}\right)_{0\le\ell\le t_{n-1}-1}
$$
Since $l-t_{t+1}+t_s+1>t_s$ for $l\ge t_{r+1}$, we know that the coefficients in the last $t_{n-1}-t_{r+1}$ columns vanish. Moreover, the coefficients in the $(t_{r+1}-1-l)$\textsuperscript{th} column is
$$
\binom{t_s}{t_s-l} = \binom{t_s}{l}
$$
for all $l\in\{0,1,\ldots,t_{r+1}-1\}$. Therefore, the submatrix corresponding to the $r+1$ rows $R_{t_r},R_{t_r+1},\ldots,R_{t_r+r}$ and the $r+1$ columns $C_{t_r},C_{t_r+1},\ldots,C_{t_r+r}$ is
$$
\matbin{t_0,t_1,\ldots,t_r}{r,r-1,\ldots,0}.
$$
Therefore, the matrix $\Md{A_n}$ is a block lower-triangular matrix with square diagonal blocks $\matbin{t_0,t_1,\ldots,t_r}{r,r-1,\ldots,0}$, for all $r\in\{0,1,\ldots,n-2\}$, of size $1,2,\ldots,n-1$.
\end{proof}

\begin{proof}[Proof of Theorem~\ref{thm2}]
From Lemma~\ref{lem2}, the matrix $\Md{A_n}$ is a block lower-triangular matrix with square diagonal blocks $\matbin{t_0,t_1,\ldots,t_r}{r,r-1,\ldots,0}$, for all $r\in\{0,1,\ldots,n-2\}$.

Using Proposition~\ref{prop2}, we have
$$
\det\matbin{t_0,t_1,\ldots,t_{r}}{r,r-1,\ldots,0} = (-1)^{\left\lfloor\frac{r+1}{2}\right\rfloor}\det\matbin{t_0,t_1,\ldots,t_{r}}{0,1,\ldots,r} = (-1)^{\left\lfloor\frac{r+1}{2}\right\rfloor}\prod_{i=1}^{r}{(2i-1)}^{r+1-i},
$$
for all $r\in\{0,1,\ldots,n-2\}$. It follows that
$$
\det\matbin{t_0,t_1,\ldots,t_{r}}{r,r-1,\ldots,0} \equiv 1\pmod{2}
$$
and thus the matrix $\matbin{t_0,t_1,\ldots,t_{r}}{r,r-1,\ldots,0}$ is invertible modulo $2$, for all $r\in\{0,1,\ldots,n-2\}$. Since all the diagonal blocks are full rank, we deduce that the matrix $\Md{A_n}$ is full rank and is invertible modulo $2$. We conclude from Proposition~\ref{prop3} that $A_n$ is a generating index set of $\setST{t_{n-1}}$.
\end{proof}

We are now ready to prove Delahan's induced-universality theorem.

\begin{proof}[Proof of Theorem~\ref{thm1}]
Let $n$ be a positive integer. We consider the $n$-subset
$$
W_{n} = \left\{ t_i+1\ \middle|\ i\in\{0,1,\ldots,n-1\} \right\}
$$
of $\left\{1,2,\ldots,t_{n-1}+1\right\}$. For any $G\in\mathcal{G}_{t_{n-1}+1}$, if $\left(a_{i,j}\right)_{1\le i,j\le t_{n-1}+1}$ is its adjacency matrix and
$$
\nabla_1=\left(a_{i+1,i+j+2}\right)_{(i,j)\in\Tn{t_{n-1}}}
$$
its upper-triangular part, then the adjacency matrix of $G[W_n]$ is $\left(a_{t_i+1,t_j+1}\right)_{0\le i,j\le n-1}$ and its upper-triangular part is
$$
\nabla_2=\left(a_{t_i+1,t_{i+j+1}+1}\right)_{(i,j)\in\Tn{n}}.
$$
The triangle $\nabla_2$ can then be seen as the subtriangle of $\nabla_1=\left(b_{i,j}\right)_{(i,j)\in\Tn{t_{n-1}}}$ defined by
$$
\nabla_2 = \left(b_{i,j}\right)_{(i,j)\in A_n},
$$
where the index set is given by $A_n\subset\Tn{t_{n-1}}$. Let $\gamma$ be the isomorphism that assigns to each triangle of zeroes and ones of size $\Tn{n}$ the simple graph of $\mathcal{G}_n$ whose adjacency matrix has this triangle as upper-triangular part and let $\theta$ be the canonical isomorphism
$$
\theta :
\begin{array}[t]{ccc}
\setSG{t_{n-1}+1} & \longrightarrow & \setST{t_{n-1}}\\
G(S) & \longmapsto & \ST{S}\\
\end{array}
$$
It is clear that
$$
\gamma\circ\pi_{A_n}\circ\theta (G) = G[W_n],
$$
for all $G\in\setSG{t_{n-1}+1}$. Finally, since $A_n$ is a generating index set of $\setST{t_{n-1}}$ by Theorem~\ref{thm2}, we deduce that $\pi_{A_n}$ and thus $\gamma\circ\pi_{A_n}\circ\theta$ both are isomorphisms.
\end{proof}

\addcontentsline{toc}{section}{References}
\bibliographystyle{plain}
\bibliography{../../Biblio_ALL}

\begin{thebibliography}{10}

\bibitem{Augier:2008aa}
Maxime Augier and Shalom Eliahou.
\newblock Parity-regular {S}teinhaus graphs.
\newblock {\em Math. Comp.}, 77(263):1831--1839, 2008.

\bibitem{Bailey:1988aa}
Craig Bailey and Wayne Dymacek.
\newblock Regular {S}teinhaus graphs.
\newblock {\em Congr. Numer.}, 66:45--47, 1988.
\newblock Nineteenth Southeastern Conference on Combinatorics, Graph Theory,
  and Computing (Baton Rouge, LA, 1988).

\bibitem{Barbe:2004aa}
Andr\'{e} Barb\'{e} and Fritz von Haeseler.
\newblock The {P}ascal matroid as a home for generating sets of cellular
  automata configurations defined by quasigroups.
\newblock {\em Theoret. Comput. Sci.}, 325(2):171--214, 2004.

\bibitem{Barbe:2009aa}
Andr\'{e} Barb\'{e} and Fritz von Haeseler.
\newblock Frame cellular automata: configurations, generating sets and related
  matroids.
\newblock {\em Discrete Math.}, 309(6):1222--1254, 2009.

\bibitem{Chang:1999aa}
Gerard~J. Chang, Bhaskar DasGupta, Wayne~M. Dym\`a\v{c}ek, Martin F\"{u}rer,
  Matthew Koerlin, Yueh-Shin Lee, and Tom Whaley.
\newblock Characterizations of bipartite {S}teinhaus graphs.
\newblock {\em Discrete Math.}, 199(1-3):11--25, 1999.

\bibitem{Chappelon2008}
Jonathan Chappelon.
\newblock On a problem of {M}olluzzo concerning {S}teinhaus triangles in finite
  cyclic groups.
\newblock {\em Integers}, 8(1):A37, 2008.
\newblock 29 pages.

\bibitem{Chappelon2009}
Jonathan Chappelon.
\newblock Regular {S}teinhaus graphs of odd degree.
\newblock {\em Discrete Math.}, 309(13):4545--4554, 2009.

\bibitem{Chappelon2011}
Jonathan Chappelon.
\newblock A universal sequence of integers generating balanced {S}teinhaus
  figures modulo an odd number.
\newblock {\em J. Combin. Theory Ser. A}, 118(1):291--315, 2011.

\bibitem{Chappelon2017a}
Jonathan Chappelon.
\newblock Periodic balanced binary triangles.
\newblock {\em Discrete Math. Theor. Comput. Sci.}, 19(3):Paper 13, 2017.

\bibitem{Chappelon2019}
Jonathan Chappelon.
\newblock Symmetric binary {S}teinhaus triangles and parity-regular {S}teinhaus
  graphs.
\newblock {\em J. Combin. Theory Ser. A}, 187:Paper No. 105561, 56 pages, 2022.

\bibitem{Chappelon2025}
Jonathan Chappelon.
\newblock Balanced {S}teinhaus triangles.
\newblock arXiv:2508.05159, 109 pages, 2025.

\bibitem{Delahan:1998aa}
Franz~A. Delahan.
\newblock Induced embeddings in {S}teinhaus graphs.
\newblock {\em J. Graph Theory}, 29(1):1--9, 1998.

\bibitem{Dymacek:1979aa}
Wayne~M. Dymacek.
\newblock Steinhaus graphs.
\newblock In {\em Proceedings of the {T}enth {S}outheastern {C}onference on
  {C}ombinatorics, {G}raph {T}heory and {C}omputing ({F}lorida {A}tlantic
  {U}niv., {B}oca {R}aton, {F}la., 1979)}, Congress. Numer., XXIII--XXIV, pages
  399--412. Utilitas Math., Winnipeg, Man., 1979.

\bibitem{Dymacek:1986aa}
Wayne~M. Dymacek.
\newblock Bipartite {S}teinhaus graphs.
\newblock {\em Discrete Math.}, 59(1-2):9--20, 1986.

\bibitem{Dymacek:1995aa}
Wayne~M. Dymacek and Tom Whaley.
\newblock Generating strings for bipartite {S}teinhaus graphs.
\newblock {\em Discrete Math.}, 141(1-3):95--107, 1995.

\bibitem{Dymacek:1999aa}
Wayne~M. Dym\`a\v{c}ek, Matthew Koerlin, and Tom Whaley.
\newblock A survey of {S}teinhaus graphs.
\newblock In {\em Combinatorics, graph theory, and algorithms, {V}ol. {I}, {II}
  ({K}alamazoo, {MI}, 1996)}, pages 313--323. New Issues Press, Kalamazoo, MI,
  1999.

\bibitem{Dymacek:2000aa}
Wayne~M. Dym\`a\v{c}ek, Jean-Guy Speton, and Tom Whaley.
\newblock Planar {S}teinhaus graphs.
\newblock In {\em Proceedings of the {T}hirty-first {S}outheastern
  {I}nternational {C}onference on {C}ombinatorics, {G}raph {T}heory and
  {C}omputing ({B}oca {R}aton, {FL}, 2000)}, volume 144, pages 193--206, 2000.

\bibitem{Eliahou:2004aa}
Shalom Eliahou and Delphine Hachez.
\newblock On a problem of {S}teinhaus concerning binary sequences.
\newblock {\em Experiment. Math.}, 13(2):215--229, 2004.

\bibitem{Gessel1985}
Ira Gessel and G\'erard Viennot.
\newblock Binomial determinants, paths, and hook length formulae.
\newblock {\em Adv. in Math.}, 58(3):300--321, 1985.

\bibitem{Harborth:1972aa}
Heiko Harborth.
\newblock Solution of {S}teinhaus's problem with plus and minus signs.
\newblock {\em J. Combinatorial Theory Ser. A}, 12:253--259, 1972.

\bibitem{Molluzzo:1978aa}
John~C. Molluzzo.
\newblock Steinhaus graphs.
\newblock In {\em Theory and applications of graphs ({P}roc. {I}nternat.
  {C}onf., {W}estern {M}ich. {U}niv., {K}alamazoo, {M}ich., 1976)}, volume 642
  of {\em Lecture Notes in Math.}, pages 394--402. Springer, Berlin, 1978.

\bibitem{Steinhaus:1964aa}
Hugo Steinhaus.
\newblock {\em One hundred problems in elementary mathematics}.
\newblock Basic Books, Inc., Publishers, New York, 1964.
\newblock With a foreword by Martin Gardner.

\end{thebibliography}

\end{document}